\documentclass[11pt]{article}
\usepackage{amsmath,amssymb,amsfonts,amsthm,epsfig}
\usepackage[usenames,dvipsnames]{xcolor}
\usepackage{bm,xspace}
\usepackage{tcolorbox}
\usepackage{cancel}
\usepackage{fullpage}
\usepackage{liyang}
\usepackage{framed}
\usepackage{verbatim}
\usepackage{enumitem}
\usepackage{array}
\usepackage{multirow}
\usepackage{afterpage}
\usepackage{mathrsfs}
\usepackage{pifont} 
\usepackage{chngpage}
\usepackage[normalem]{ulem}
\usepackage{boxedminipage}
\usepackage{caption}
\usepackage{forest}

\def\colorful{0}

\ifnum\colorful=1
\newcommand{\violet}[1]{{\color{violet}{#1}}}

\fi
\ifnum\colorful=0
\newcommand{\violet}[1]{{{#1}}}

\fi

\newcommand{\error}{\mathrm{error}}
\newcommand{\round}{\mathrm{round}}

\newcommand{\score}{\mathrm{Score}}
\newcommand{\BuildDT}{\textsc{BuildDT}}
\newcommand{\UnbiasedEstimator}{\textsc{UnbiasedEstimator}}
\newcommand{\Reconstructor}{\textsc{Reconstructor}}

\makeatletter
\newtheorem*{rep@theorem}{\rep@title}
\newcommand{\newreptheorem}[2]{
\newenvironment{rep#1}[1]{
 \def\rep@title{#2 \ref{##1}}
 \begin{rep@theorem}\itshape}
 {\end{rep@theorem}}}
\makeatother

\newreptheorem{theorem}{Theorem}

\newcommand{\pparagraph}[1]{\bigskip \noindent {\bf {#1}}}

\begin{document}

\title{\violet{Reconstructing Decision Trees} \vspace{15pt} 
}


\author{Guy Blanc \vspace{8pt} \\ \hspace{-5pt}{\sl Stanford} \and \hspace{10pt} Jane Lange \vspace{8pt} \\
\hspace{4pt}  {\sl MIT}
\and Li-Yang Tan \vspace{8pt} \\ \hspace{-8pt} {\sl Stanford}}

\date{\vspace{5pt}\small{\today}}

\maketitle


\begin{abstract} 
We give the first {\sl reconstruction algorithm} for decision trees: given queries to a function~$f$ that is $\opt$-close to a size-$s$ decision tree, our algorithm provides query access to a decision tree $T$ where:
\begin{itemize}
\item[$\circ$] $T$ has size $S \coloneqq s^{O((\log s)^2/\eps^3)}$;
\item[$\circ$] $\dist(f,T)\le O(\opt)+\eps$;
\item[$\circ$] Every query to $T$ is answered with $\poly((\log s)/\eps)\cdot \log n$ queries to~$f$ and in $\poly((\log s)/\eps)\cdot n\log n$ time. 
\end{itemize} 

This yields a {\sl tolerant tester} that distinguishes functions that are close to size-$s$ decision trees from those that are far from size-$S$ decision trees. The polylogarithmic dependence on $s$ in the efficiency of our tester is exponentially smaller than that of existing testers.  

Since decision tree complexity is well known to be related to numerous other boolean function properties, our results also provide a new algorithms for reconstructing and testing these properties. 



\end{abstract}

\thispagestyle{empty}

\newpage
\setcounter{page}{1}

\section{Introduction} 

\violet{We study the problem of {\sl reconstructing} decision trees: given queries to a function $f$ that is close to a size-$s$ decision tree, provide fast query access to a decision tree, ideally one of size not much larger than $s$, that is close to~$f$.  This can be viewed as an ``on the fly" variant of the problem of properly and agnostically learning  decision trees, where the goal there is to output the entire decision tree hypothesis.  More broadly, reconstruction algorithms, introduced by Ailon, Chazelle, Comandur, and Liu~\cite{ACCL08}, can be viewed as sublinear algorithms that restore structure---in our case, that of a decision tree---in a function that has been lost due to noise.}



Decision trees have long been a popular and effective model in machine learning, and relatedly, they are among the most intensively studied concept classes in learning theory.  The literature on learning decision trees is vast, spanning three decades and studying the problem in a variety of models and from a variety of perspectives~\cite{EH89,Riv87,Blu92,Han93,Bsh93,BFJKMR94,HJLT96,MR02,JS06,OS07,KS06,KST09,HKY18,CM19,BDM20,BLT-ITCS,BLT-ICML,BGLT-NeurIPS1,BLT21icalp,BLQT21}. 
\violet{In contrast, the problem of reconstructing decision trees has thus far been surprisingly understudied.} 




\subsection{Our contributions}

\violet{We give the first reconstruction algorithm for decision trees.  Our algorithm achieves a {\sl polylogarithmic} dependence on $s$ in its query and time complexities, exponentially smaller than the information-theoretic minimum required to learn.}


\begin{theorem}[Main result]
\label{thm:reconstruct}
There is a randomized algorithm which, 
given queries to $f : \bn\to\bits$ and parameters $s\in \N$ and $\eps\in (0,1)$, provides query access to a fixed decision tree~$T$ where 
\begin{itemize}
\item[$\circ$] $T$ has size $s^{O((\log s)^2/\eps^3)}$;
\item[$\circ$] $\dist(T,f) \le O(\opt_s) + \eps$ w.h.p., where $\opt_s$ denotes the distance of $f$ to the closest size-$s$ decision tree;
\item[$\circ$] Every query to $T$ is answered with $\poly((\log s)/\eps)\cdot \log n$ queries to~$f$ and in $\poly((\log s)/\eps)\cdot n\log n$ time. 
\end{itemize} 
\end{theorem}

\violet{Notably, in the standard setting where $s = \poly(n)$, the query and time complexities of our algorithm are $\polylog(n)$ and $\tilde{O}(n)$ respectively.  Previously, the only known approach was to simply properly and agnostically learn $f$; the current fastest such algorithm has query and time complexities $n^{O(\log\log n)}$~\cite{BLQT21}.} 

Our reconstruction algorithm is furthermore {\sl local} in the sense of Saks and Seshadhri~\cite{SS10}, allowing queries to be answered in parallel assuming a shared random string.  In particular, once $f, s, \eps$ and the random string are fixed, all queries are answered consistently with a single decision tree.

\subsubsection{\violet{Implications of~\Cref{thm:reconstruct} and further results}}
By a standard reduction,~\Cref{thm:reconstruct} gives a {\sl tolerant tester} for decision trees:

\begin{corollary}[Tolerant testing of decision trees] 
\label{cor:main}  
There is a randomized algorithm which, given queries to $f : \bn \to\bits$ and parameters $s\in\N$ and $\eps \in (0,1)$, 
\begin{itemize}
\item[$\circ$] Makes $\poly((\log s)/\eps) \cdot \log n$ queries to $f$, runs in $\poly((\log s)/\eps)\cdot n\log n$ time, and 
\item[$\circ$] Accepts w.h.p.~if $f$ is $\eps$-close to a size-$s$ decision tree; 
\item[$\circ$] Rejects w.h.p.~if $f$ is $\Omega(\eps)$-far from size-$s^{O((\log s)^2/\eps^3)}$ decision trees.
\end{itemize} 
\end{corollary} 
This adds to a long line of work on testing decision trees~\cite{KR00,DLMORSW07,CGM11,BBM12,Bsh20}.  \violet{We give an overview of prior testers in~\Cref{sec:prior}, mentioning for now that they all have (at least) an exponentially larger dependence on $s$ in their query and time complexities.}  

\paragraph{A new connection between tolerant testing and learning.} It would be preferable if our tester can be improved to reject all $f$'s that are far
from size-$s$ decision trees---or more strongly, if our reconstructor can be improved to provide query access to a size-$s$ decision tree.  

We show that such a tester, even one that is considerably less efficient than ours, would yield the first polynomial-time algorithm for properly learning decision trees:

\begin{theorem}[Tolerant testing $\Longrightarrow$ Proper learning] 
\label{thm:lower-bound-intro}
Suppose there is an algorithm which, given query access to $f : \bn \to \bits$ and parameters $s\in \N$ and $\eps \in (0,1)$, 
\begin{itemize}
\item[$\circ$] Makes $\poly(s,n,1/\eps)$ queries to $f$, runs in $\poly(s,n,1/\eps)$ time, and 
\item[$\circ$] Accepts w.h.p.~if $f$ is $\eps$-close to a size-$s$ decision tree; 
\item[$\circ$] Rejects w.h.p.~if $f$ is $\Omega(\eps)$-far from size-$s$ decision trees.
\end{itemize} 
Then there is a $\poly(s,n,1/\eps)$-time membership query algorithm for properly learning size-$s$ decision trees with respect to the uniform distribution. 
\end{theorem} 

This would represent a breakthrough on a central open problem in learning theory.  Recent work of Blanc, Lange, Qiao, and Tan~\cite{BLQT21} gives a $\poly(n)\cdot s^{O(\log\log s)}$ time algorithm, improving on the prior state of the art of $n^{O(\log s)}$~\cite{EH89}.  Neither~\cite{EH89}'s nor~\cite{BLQT21}'s algorithm goes through testing.


It is well known and easy to see that proper learning algorithms yield comparably efficient testers~\cite{GGR98}. \Cref{thm:lower-bound-intro} provides an example of a converse; \violet{we find the existence of such a converse surprising, and are not aware of any previous examples.}

\paragraph{Reconstructors and testers for other properties.} 
Decision tree complexity is quantitatively related to numerous other complexity measures of boolean functions: Fourier degree, approximate degree, randomized and quantum query complexities, certificate complexity, block sensitivity, sensitivity, etc.  Our results therefore immediately yield new reconstructors and tolerant testers for these properties.  For example, we have the following:

\begin{corollary}[Reconstruction of low Fourier degree functions]  
\label{cor:reconstruct-degree} 
There is a randomized algorithm which, 
given queries to $f : \bn\to\bits$ and parameters $d\in \N$ and $\eps\in (0,1)$, provides query access to a fixed function $g : \bn\to\bits$ where 
\begin{itemize}
\item[$\circ$] $g$ has Fourier degree $O(d^7/\eps^2)$,
\item[$\circ$] $\dist(f,g) \le O(\opt_d) + \eps$ w.h.p., where $\opt_d$ denotes the distance of $f$ to closest $h : \bn\to\bits$ of Fourier degree $d$.
\item[$\circ$] Every query to $g$ is answered in $\poly(d,1/\eps)\cdot n\log n$ time and with $\poly(d,1/\eps)\cdot \log n$ queries to~$f$.
\end{itemize} 
\end{corollary} 

This  in turn yields a tolerant tester for Fourier degree. As in the case for decision trees, all prior testers for low Fourier degree~\cite{DLMORSW07,CGM11,CGM11-soda,BBM12,BH13,Bsh20} have an exponential dependence on $d$ in their query and time complexities. 

\Cref{table:other-measures} lists examples of measures for which we obtain new reconstruction algorithms, each of which in turn give new tolerant testers.


\vspace{5pt}
\begin{table}[h!]
\begin{adjustwidth}{-1in}{-1in}
\renewcommand{\arraystretch}{1.8}
\centering
\begin{tabular}{|c|c|c|}
\hline
\multirow{3}{*}{\sl Complexity measure}   & {\sl Assumption} & {\sl Guarantee}  \\ [-.5em] 
&~~Query access to $f$ that is~~&~~Query access to $g$ that is~~\\ [-.9em] 
&~~$\opt_d$-close to $h$ where:~~&~~$O(\opt_d+\eps)$-close to $f$ where:~~\\ [.2em] \hline \hline
Fourier degree & $\deg(h) \le d$ &~~$\deg(g) \le O(d^7/\eps^2)$~~ \\ [.2em] \hline 
Approximate degree &~~$\wt{\deg}(h) \le d$~~&~~$\wt{\deg}(g) \le O(d^{9}/\eps^2)$~~ \\ [.2em] \hline 
~~Randomized query complexity~~&~~$\mathrm{R}(h) \le d$~~&~~$\mathrm{R}(g) \le O(d^7/\eps^2)$~~ \\ [.2em] \hline 
~~Quantum query complexity~~&~~$\mathrm{Q}(h) \le d$~~&~~$\mathrm{Q}(g) \le {O(d^{10}/\eps^2)}$~~ \\ [.2em] \hline 
~~Certificate complexity~~&~~$\mathrm{C}(h) \le d$~~&~~$\mathrm{C}(g) \le O(d^5/\eps^2)$~~ \\ [.2em] \hline 
~~Block sensitivity~~&~~$\mathrm{bs}(h) \le d$~~&~~$\mathrm{bs}(g) \le O(d^8/\eps^2)$~~ \\ [.2em] \hline 
~~Sensitivity~~&~~$\mathrm{s}(h) \le d$~~&~~$\mathrm{s}(g) \le O(d^{13}/\eps^{2})$~~ \\ [.2em] \hline \end{tabular}
\end{adjustwidth}
  \captionsetup{width=.9\linewidth}
\caption{Performance guarantees of our reconstruction algorithms for various complexity measures. In each row, $\opt_d$ denotes the distance from $f$ to the closest function $h$ such that the complexity measure of that row for $h$ is bounded by $d$. In all cases, every query to $g$ is answered in $\poly(d,1/\eps)\cdot n\log n$ time with $\poly(d,1/\eps)\cdot \log n$ queries to $f$. }
\label{table:other-measures}
\vspace{-5pt} 
\end{table}
%
%

\subsection{Background and comparison with prior work} 
\label{sec:prior}

\violet{As already mentioned,~\Cref{thm:reconstruct} gives the first reconstruction algorithm for decision trees. The problem of testing decision trees, on the other hand, has been intensively studied. } 




\violet{\pparagraph{Testing decision trees.}}  Recent work of Bshouty~\cite{Bsh20} gives an algorithm, running in $\poly(s^s,1/\eps)\cdot n$ time and using $O((s\log s)/\eps)$ queries, that distinguishes between size-$s$ decision trees from functions that are $\eps$-far from size-$s$ decision trees.   Prior to~\cite{Bsh20},  Chakraborty, Garc{\'\i}a-Soriano, and Matsliah~\cite{CGM11} gave an $O((s\log s)/\eps^2)$-query algorithm, and before that Diakonikolas, Lee, Matulef, Onak, Rubinfeld, Servedio, and Wan~\cite{DLMORSW07} gave an $\tilde{O}(s^4/\eps^2)$-query algorithm.  Like \cite{Bsh20}'s algorithm, the algorithms of~\cite{CGM11,DLMORSW07} also run in $\poly(s^s,1/\eps)\cdot n$ time.\footnote{\violet{All these testers enjoy a weak form of tolerance: they are in fact able to distinguish between functions that are $O(\poly(\eps/s))$-close to size-$s$ decision trees from those that are $\eps$-far from size-$s$ decision trees. (Briefly, this is because their queries, while correlated, are each uniformly distributed.)}} 

Compared to these algorithms, our algorithm in~\Cref{cor:main} solves an incomparable problem with efficiency parameters that compare rather favorably with theirs.  Notably, our time and query complexities both depend {\sl polylogarithmically} on $s$ instead of exponentially and super-linearly respectively.



Turning to the parameterized setting, Kearns and Ron~\cite{KR00} gave a tester with time and query complexities $\poly(n^n, (\log s)^n)$ that distinguishes size-$s$ decision trees over $[0,1]^n$ from functions that are $(\frac1{2}-n^{-\Theta(n)})$-far from size-$\poly(2^n,s)$ decision trees.   The parameters of this result are such that one should think of the dimension `$n$' as being a constant rather than an asymptotic parameter.   


\paragraph{Property reconstruction.} 

Property reconstruction was introduced by Ailon, Chazelle, Comandur, and Liu~\cite{ACCL08}. (See also the work of Austin and Tao~\cite{AT10}, who termed such algorithms~``repair algorithms".)   Reconstruction has since been studied for a number of properties, including monotone functions~\cite{ACCL08,SS10,BGJJRW12}, hypergraph properties~\cite{AT10},  convexity~\cite{CS06}, expanders~\cite{KPS13}, Lipschitz functions~\cite{JR13}, graph connectivity and diameter~\cite{CGR13}, and error correcting codes~\cite{CFM14}.   Property reconstruction falls within the {\sl local computation algorithms} framework of Rubinfeld, Tamir, Vardi, and Xie~\cite{RTVX11}.

The paper of Blanc, Gupta, Lange, and Tan~\cite{BGLT-NeurIPS2}
 designs a decision tree learning algorithm that is amenable to {\sl learnability estimation}~\cite{KV18,BH18}:  given a training set $S$ of {\sl unlabeled} examples, the performance of this algorithm $\mathcal{A}$ trained on $S$---that is, the generalization error of the hypothesis that $\mathcal{A}$ would construct if we were to label all of $S$ and train $\mathcal{A}$ on it---can be accurately estimated by labeling only a small number of the examples in $S$.  Their techniques can be used to derive a reconstruction algorithm that achieves guarantees similar to those in~\Cref{thm:reconstruct}, but only for {\sl monotone} functions~$f$.  This limitation is inherent: as noted in~\cite{BGLT-NeurIPS2}, their algorithm is fails for non-monotone functions.  

\subsubsection{The work of Bhsouty and Haddad-Zaknoon}
\label{sec:Bshouty} 
Subsequent to the posting of our work to the ArXiv,  Bshouty and Haddad-Zaknoon~\cite{BHZ21} have given a tester that is closely related, but incomparable, to~\Cref{cor:main}.  Their tester: 
\begin{itemize}
\item[$\circ$] Makes $\poly(s,1/\eps)$ queries to $f$, runs in $\poly(n,1/\eps)$ time, and 
\item[$\circ$] Accepts w.h.p.~if $f$ is exactly a size-$s$ decision tree; 
\item[$\circ$] Rejects w.h.p.~if $f$ is $\eps$-far from size-$(s/\eps)^{O(\log(s/\eps))}$ decision trees.
\end{itemize} 

Comparing~\cite{BHZ21}'s tester to ours, their query complexity is independent of~$n$ (whereas ours has a $\log n$ dependence), and the size of decision trees in their reject condition is only $(s/\eps)^{O(\log(s/\eps))}$ (whereas we require $s^{O((\log s)^2/\eps^3)}$). 

On the other hand, our tester is tolerant and has query complexity that achieves a polylogarithmic instead of polynomial dependence on $s$.  Furthermore,~\cite{BHZ21} does not give a reconstruction algorithm, \violet{while that is the main contribution of our work.}

\subsection{Future directions} 
\label{sec:future}


We list a few concrete avenues for future work suggested by our results:


\begin{itemize}


\item[$\circ$] {\sl Tighter connections between testing and learning:} Our tester rejects functions that are $\Omega(\eps)$-far from $\mathrm{quasipoly}(s)$ decision trees, and~\Cref{thm:lower-bound-intro} shows that a tester that rejects functions that are $\Omega(\eps)$-far from size-$s$ decision trees would yield a comparably efficient algorithm for properly learning decision trees.  A concrete avenue for future work is to narrow this gap between $\mathrm{quasipoly}(s)$ and $s$, with the ultimate goal of getting them to match.   

There are also other ways in which \Cref{thm:lower-bound-intro} could be strengthened: Do {\sl non-tolerant} testers for decision trees yield proper learning algorithms?  Do tolerant testers yield proper learning algorithms with {\sl agnostic} guarantees?  


\item[$\circ$] {\sl Improved reconstruction algorithms and testers for other properties:} The reconstruction algorithms that we obtain for the properties listed in~\Cref{table:other-measures} follow by combining~\Cref{thm:reconstruct} with known relationships between these measures and decision tree complexity.  It would be interesting to obtain improved parameters by designing reconstruction algorithms that are tailored to each of these properties, without going through decision trees.  

The same questions can be asked of property testers, and about properties that are not known to be quantitatively related to decision tree size. Can we achieve similar exponential improvements in the time and query complexities of non-parameterized testers by relaxing to the parameterized setting? \Cref{thm:lower-bound-intro} could be viewed as suggesting that for certain properties, efficient algorithms may only be possible in the parameterized setting. 
\end{itemize} 

Finally, we mention that there remains a large gap in the known bounds on the query complexity of non-tolerant testing of decision trees in the non-parameterized setting: the current best upper bound is~$\tilde{O}(s)$~\cite{Bsh20,CGM11} whereas the current best lower bound is $\Omega(\log s)$~\cite{DLMORSW07,BBM12}.  It would be interesting to explore whether our techniques could be useful in closing this exponential gap. 

\paragraph{Notation.} All probabilities and expectations are with respect to the uniform distribution unless otherwise stated; we use boldface (e.g.~$\bx$) to denote random variables. For two functions $f,g : \bn\to\bits$, we write $\dist(f,g)$ to denote the quantity $\Pr[f(\bx)\ne g(\bx)]$.  We say that $f$ and $g$ are {\sl $\eps$-close} if $\Pr[f(\bx)\ne g(\bx)] \le \eps$, and {\sl $\eps$-far} otherwise.   

For a function $f : \bn\to\bits$, a decision tree $T$ over the same variables as $f$, and a node $v$ in $T$, we write $f_v$ to denote the subfunction of $f$ obtained by restricting $f$ according to the root-to-$v$ path in $T$.   We write $|v|$ to denote the depth of $v$ within $T$, and so the probability that a uniform random $\bx\sim \bn$ reaches $v$ is $2^{-|v|}$.

\section{Proofs of~\Cref{thm:reconstruct} and~\Cref{cor:main}} 

Our proof of~\Cref{thm:reconstruct} has two main components:  

\begin{itemize}
\item[$\circ$] A structural lemma about functions $f$ that are $\opt_s$-close to a size-$s$ decision tree $T^\star$.  While we have no information about the structure of this tree $T^\star$ that $f$ is $\opt_s$-close to, we will show that $f$ is $O(\opt_s + \eps)$-close to a tree $T^\diamond$ of size $S = S(s,\eps)$ with a very specific structure.  
 
\item[$\circ$] An algorithmic component that leverages this specific structure of $T^\diamond$ to show that for any input $x \in\bn$, the value of $T^\diamond(x)$ can be computed with only $\log S \cdot \log n$ queries to $f$.  
\end{itemize} 


\Cref{sec:structural} will be devoted to the structural lemma and~\Cref{sec:algorithmic} to the algorithmic component.  We prove~\Cref{thm:reconstruct} in~\Cref{sec:proof-of-reconstruct}, and we derive~\Cref{cor:main} as a simple consequence of~\Cref{thm:reconstruct} in~\Cref{sec:proof-of-tester}.

\subsection{Structural component of~\Cref{thm:reconstruct}}
\label{sec:structural} 

\begin{definition}[Noise sensitivity]
The {\sl noise sensitivity of $f : \bn\to\bits$ at noise rate $p$} is the quantity 
\[ \NS_p(f) \coloneqq \Pr[f(\bx) \ne f(\by)],\]
where $\bx\sim \bn$ is uniform random and $\by\sim_p\bx$ is a {\sl $p$-noisy copy of $\bx$}, obtained from $\bx$ by independently rerandomizing each coordinate with probability $p$. 
\end{definition} 

We assign each coordinate $i\in [n]$ of a function $f$ a {\sl score}, which measures the expected decrease in the noise sensitivity of $f$ if $x_i$ is queried:  

\begin{definition}[Score of a variable]
\label{def:score}
    Given a function $f:\bn \to \bits$, noise rate $p \in (0,1)$, and coordinate $i \in [n]$, the {\sl score} of $x_i$ is defined as
\[        \score_i(f, p) = \NS_{p}(f) - \Ex_{\bb\in\bits}\big[\NS_{p}(f_{x_i = \bb})\big].
\] 
\end{definition}

(Our notion of score is equivalent, up to scaling factors depending on $p$, to the notion of ``noisy influence" as in defined in O'Donnell's monograph~\cite{ODbook}. We use our definition of score as it simplifies our presentation.)  We are now ready to define the tree $T^\diamond$ described at the beginning of this section and state our structural lemma. 

\begin{definition} 
\label{def:top-down-tree} 
For a function $f : \bn\to\bits$, parameters $d\in \N$ and $p \in (0,1)$, we write $T^{d,p}_f$ to denote the complete decision tree of depth $d$ defined as follows: 
\begin{itemize} 
\item[$\circ$]  At every internal node $v$, query $x_i$ where $i\in [n]$ maximizes $\score_i(f_v,p)$.\footnote{Ties are arbitrarily broken; our results hold regardless of how ties are broken.}
\item[$\circ$]  Label every leaf $\ell$ with $\sign(\E[f_\ell])$. 
\end{itemize} 
\end{definition}

\begin{lemma}[Structural lemma]
\label{lemma:BGLT} 
Let $f : \bn\to\bits$ be $\opt_s$-close to a size-$s$ decision tree.  Then for $d = O((\log s)^3/\eps^3)$ and $p = \eps/(\log s)$, we have $\dist(f,T^{d,p}_f) \le O(\opt_s)+\eps$. 
\end{lemma}

\section{Proof of~\Cref{lemma:BGLT}} 
\label{appendix} 
\paragraph{Noise-sensitivity-based potential function.} First, we introduce the potential function that will facilitate our proof of~\Cref{lemma:BGLT}. Every decision tree $T$ naturally induces a distribution over its leaves where each leaf $\ell$ receives weight $2^{-|\ell|}$.  We write $\bell\sim T$ to denote a draw of a leaf of $T$ according to this distribution. 

\begin{definition}[Noise sensitivity of $f$ with respect to a tree $T$] 
\label{def:NS-wrt-T} 
Let $f :\bn\to\bits$ be a function, $p\in (0,1)$, and $T$ be a decision tree.  The {\sl noise sensitivity of $f$ at noise rate $p$ with respect to $T$} is the quantity 
\[ \NS_p(f,T) \coloneqq \Ex_{\bell\sim T} \big[ \NS_p(f_{\bell})\big].\] 
\end{definition} 
Note that if $T$ is the empty tree, then $\NS_p(f,T)$ is simply $\NS_p(f)$, the noise sensitivity of $f$ at noise rate $p$.  The following proposition is a bound on $\NS_p(f)$ that takes into account its distance from a small decision tree:  

\begin{proposition}[Noise sensitivity of $f$] 
\label{prop:NS-of-f} 
Let $f : \bn\to\bits$ be $\opt_s$-close to a size-$s$ decision tree $T$.  For all $p\in (0,1)$, we have $\NS_p(f) \le p\log s + 2\,\opt_s.$
\end{proposition} 

\begin{proof} 
Let $\bx\sim\bn$ be uniform random, $\by\sim_p \bx$ be a $p$-noisy copy of $\bx$, and $\bx^{\oplus i}$ denote $\bx$ with its $i$-th coordinate flipped.  We first observe that 
\begin{align*} 
\NS_p(f) &= \Pr[f(\bx)\ne f(\by)] \\
&\le \Pr[f(\bx)\ne T(\bx)] + \Pr[T(\bx) \ne T(\by)] + \Pr[T(\by) \ne f(\by)] \\
&= \NS_p(T) + 2\,\opt_s. 
\end{align*} 
To bound $\NS_p(T)$, we use the inequality $\NS_p(T) \le p\cdot \Inf(T)$ where $\Inf(T) \coloneqq \sum_{i=1}^n \Pr[f(\bx) \ne f(\bx^{\oplus i})]$ is the total influence of $T$~\cite[Exercise 2.42]{ODbook}, along with the bound $\Inf(T) \le \log s$ (see e.g.~\cite{OS07}). 
\end{proof} 

%
%
%
%
%

We prove~\Cref{lemma:BGLT} by quantifying the difference between $\NS_p(f,T^{j+1,p}_f)$ and $\NS_p(f,T^{j,p}_f)$: we show that for every $j\in\N$, either $\dist(f,T^{j,p}_f)\le O(\opt_s + \eps)$ or it must be the case that $\NS_p(f,T^{j+1,p}_f)$ is significantly smaller than $\NS_p(f,T^{j,p}_f)$.  Since $\NS_p(f,T) \ge 0$ for all trees~$T$, the second case can only happen so many times before we fall into the first case.  

We will need the following result from~\cite{OSSS05}: 


\begin{theorem}[Theorem 3.2 of~\cite{OSSS05}] 
\label{thm:OSSS-non-buggy} 
Let $T : \bn \to \bits$ be a decision tree.  For all functions $g : \bn\to \R$, writing $\bx,\bx'\sim\bn$ to denote uniform random and independent inputs and $\bx^{\sim i}$ to denote $\bx$ with its $i$-th coordinate rerandomized, 
\[ \mathrm{CoVr}(T,g) \le \sum_{i=1}^n \lambda_i(T) \cdot \Ex_{\bx} \big[|g(\bx)-g(\bx^{\sim i})|\big],\] 
where 
\begin{align*}
\mathrm{CoVr}(T,f) &\coloneqq \Ex_{\bx,\bx'}\big[|T(\bx)-g(\bx')|\big] - \Ex_{\bx}\big[|T(\bx)-g(\bx)|\big], \\
\lambda_i(T) &\coloneqq \Pr[\,\text{$T$ queries $\bx_i$}\,].
\end{align*}
\end{theorem}

By applying~\Cref{thm:OSSS-non-buggy} to a suitably smoothened version of $f$, we are able to derive a lower bound on the score of the highest-scoring variable of $f$.  

\begin{definition}[$p$-smoothed version of $f$]
For $f : \bn\to\bits$ and $p\in (0,1)$, the {\sl $p$-smoothed version of $f$} is the function $\tilde{f}^{(p)} : \bn\to [-1,1]$, 
\[ \tilde{f}^{(p)}(x) = \Ex_{\by\sim_p x}[f(\by)] = \sum_{S\sse [n]} (1-p)^{|S|} \wh{f}(S)\prod_{i\in S}x_i,\] 
where the $\wh{f}(S)$ is the $S$-th Fourier coefficients of $f$.  When $p$ is clear from context, we write $\tilde{f}$.
\end{definition}

\begin{lemma}[Score of the highest-scoring variable]
\label{lem:noisy-OSSS-truncated-agnostic-fixed}
Let $f : \bits^n \to \bits$ be a function, $p\in (0,1)$, and $\tilde{f} = \tilde{f}^{(p)}$ be its $p$-smoothed version.  For all size-$s$ decision trees $T : \bn\to \bits$, 
\[ \max_{i\in [n]}\big\{\sqrt{\score_i(f,p)}\big\} \ge \frac{\sqrt{p}}{\log s} \cdot \Big(\lfrac1{2}\Var(\tilde{f}) - \E\big[|T(\bx)-\tilde{f}(\bx)|\big]\Big).\]
\end{lemma} 
  
\begin{proof} 
Applying~\Cref{thm:OSSS-non-buggy} with `$g$' being the $p$-smoothed version $\tilde{f}$ of $f$, we have 
\begin{equation}
\label{eq:OSSS-non-buggy-for-us} \mathrm{CoVr}(T,\tilde{f}) \le \sum_{i=1}^n \lambda_i(T) \cdot \E[|\tilde{f}(\bx)-\tilde{f}(\bx^{\sim i})|].\end{equation} 
We first lowerbound the LHS of~\Cref{eq:OSSS-non-buggy-for-us}.  For $\bx,\bx'\sim \bn$ uniform and independent,
\begin{align} 
\mathrm{CoVr}(T,\tilde{f}) &= \E\big[|T(\bx)-\tilde{f}(\bx')|\big] - \E\big[|T(\bx)-\tilde{f}(\bx)|\big] \tag*{(Definition of $\mathrm{CoVr}$)} \nonumber \\
&\ge \E\big[|\tilde{f}(\bx)-\tilde{f}(\bx')|\big] - 2\E\big[|T(\bx)-\tilde{f}(\bx)|\big] \tag*{(Triangle inequality)}\nonumber \\
&\ge \lfrac1{2}\E\big[(\tilde{f}(\bx)-\tilde{f}(\bx'))^2\big] - 2\E\big[|T(\bx)-\tilde{f}(\bx)|\big] \tag*{($\tilde{f}$ is $[-1,1]$-valued)}\nonumber \\
&\ge \Var(\tilde{f}) - 2\E\big[|T(\bx)-\tilde{f}(\bx)|\big]. \label{eq:cov-lb-fixed}  
\end{align} 
For a function $g : \bn\to\R$, its {\sl $i$-th discrete derivative} is the function 
\[ (D_ig)(x) \coloneqq \lfrac1{2}\big(g(x^{i=1})-g(x^{i=-1})\big)  = \ds\sum_{S\ni i} \wh{g}(S)\prod_{j \in S \setminus \{i\}} x_j,\]
where $x^{i=b}$ denotes $x$ with its $i$-th coordinate set to $b$.  With this definition in hand, we now analyze the expectation on the RHS of~\Cref{eq:OSSS-non-buggy-for-us}.  By Jensen's inequality, \[ \E[|\tilde{f}(\bx)-\tilde{f}(\bx^{\sim i})|]^2 \le \Ex_{\bx}\big[ (\tilde{f}(\bx)-\tilde{f}(\bx^{\sim i}))^2\big] = \lfrac1{2}\ds\Ex_{\bx}\big[ (\tilde{f}(\bx^{i=1})-\tilde{f}(\bx^{i=-1}))^2\big] = 2\Ex_{\bx}\big[D_i\wt{f}(\bx)^2\big]. \] 
Applying Plancherel's identity twice, 
\[ \Ex_{\bx}\big[ D_i\tilde{f}(\bx)^2\big] = \sum_{S \ni i} (1-p)^{2|S|}\wh{f}(S)^2 \le \sum_{S\ni i} (1-p)^{|S|}\wh{f}(S)^2 = \Ex_{\bx,\by}\big[ D_if(\bx)D_if(\by)\big] \] 
where $\by\sim_p\bx$ is a $p$-noisy copy of $\bx$.  It follows from a straightforward calculation~\cite[Lemma 3.2]{BGLT-NeurIPS1} that 
\[ \Ex_{\bx,\by}\big[ D_if(\bx)D_if(\by)\big] = \frac{2\cdot  \score_i(f,p)}{p}.\]
Therefore, combining the three equations above we have shown that 
\begin{equation} \label{eq:score-non-negative} \Ex_{\bx}\big[ |\tilde{f}(\bx)-\tilde{f}(\bx^{\sim i})|\big] \le \sqrt{\frac{4\cdot  \score_i(f,p)}{p}}.
\end{equation} 
Plugging this inequality into the RHS of~\Cref{eq:OSSS-non-buggy-for-us}, 
\begin{align} 
\sum_{i=1}^n \lambda_i(T) \cdot \ds\Ex_{\bx}\big[|\tilde{f}(\bx)-\tilde{f}(\bx^{\sim i})|\big] &\le  \sum_{i=1}^n \lambda_i(T) \cdot \sqrt{\frac{4\cdot  \score_i(f,p)}{p}} \nonumber \\
&\le \ds \max_{i\in [n]} \big\{ \sqrt{\score_i(f,p)}\big\}  \cdot \frac{2}{\sqrt{p}} \cdot  \sum_{i=1}^n \lambda_i(T) \nonumber \\
&\le  \max_{i\in [n]} \big\{ \sqrt{\score_i(f,p)}\big\}  \cdot \frac{2\log s}{\sqrt{p}}, \label{eq:RHS-ub-fixed} 
 \end{align}  
 where the final inequality holds because 
 \[ \sum_{i=1}^n \lambda_i(T) = \sum_{i=1}^n \Pr[\,\text{$T$ queries $\bx_i$}\,] = \Ex_{\bell\sim T}\big[|\bell|\big]  \le \log s. \] 
The lemma follows by combining~\Cref{eq:OSSS-non-buggy-for-us,eq:cov-lb-fixed,eq:RHS-ub-fixed}.
\end{proof}

\subsubsection{Proof of~\Cref{lemma:BGLT}} 

Let $T^\star$ be the size-$s$ decision tree that $f$ is $\opt_s$-close to. 
Fix $j\in \N$ and consider the tree $T^{j,p}_f$.  We have that: 
\begin{align*}
\NS_p(f, T^{j+1,p}_f) &\le \NS_p(f, T^{j,p}_f) - \Ex_{\bell\sim T^{j,p}_f}\Big[\max_{i\in [n]}\big\{\score_i(f_{\bell}, p)\big\}\Big]  \tag*{(\Cref{def:score})} \\
&\le \NS_p(f, T^{j,p}_f) - \bigg(\Ex_{\bell\sim T^{j,p}_f}\Big[\max_{i\in [n]}\big\{\sqrt{\score_i(f_{\bell}, p)}\big\}\Big]\bigg)^2. \tag*{(Jensen's inequality)}
\end{align*}


Recall that we write $\bell\sim T$ to denote a draw of a leaf of $T$ where each leaf $\ell$ receives weight $2^{-|\ell|}$.   
We consider two cases:\medskip

\noindent {\bf{Case 1:}} $\Ex_{\bell \sim T^{j,p}_f}[\Var(\wt{f_{\bell}})] \ge 2\,(\Ex_{\bell,\bx}\big[|{T}^\star(\bx)-\wt{f_{\bell}}(\bx)|\big] + \eps)$.\medskip 

 In this case we apply~\Cref{lem:noisy-OSSS-truncated-agnostic-fixed} to each leaf $\ell$ of $T^{j,p}_f$ to get that 
 \begin{align*} \Ex_{\bell\sim T^{j,p}_f}\Big[\max_{i\in [n]}\big\{\sqrt{\score_i(f_{\bell}, p)}\big\}\Big] &\ge \frac{\sqrt{p}}{\log s} \cdot \Ex_{\bell\sim T^{j,p}_f}\Big[\lfrac1{2}\Var(\tilde{f_{\bell}}) - \E\big[|T^\star(\bx)-\tilde{f_{\bell}}(\bx)|\big]\Big]  \\
 &\ge  \frac{\eps\sqrt{p}}{\log s},
 \end{align*} 
 and hence
\begin{align*} \NS_p(f,T^{j+1,p}_f) &\le \NS_p(f,T^{j,p}_f) - \frac{\eps^2 p}{(\log s)^2}. \\
&= \NS_p(f,T^{j,p}_f)-\frac{\eps^3}{(\log s)^3}. \tag*{(Our choice of $p = \eps/(\log s)$)} 
\end{align*}

%

\medskip

\noindent {\bf{Case 2:} $\Ex_{\bell \sim T^{j,p}_f}[\Var(\wt{f_{\bell}})] < 2\,(\Ex_{\bell,\bx}\big[|{T}^\star (\bx)-\wt{f_{\bell}}(\bx)|\big] + \eps$).}\medskip

 In this case we claim that $\dist(f,T^{j,p}_f) \le O(\opt_s + \eps)$.  We will need a couple of simple propositions:

%

\begin{proposition}
\label{prop:NS-of-leaves}
$\Ex_{\bell,\bx}[(\wt{f_{\bell}}(\bx)-f_\ell(\bx))^2] \le 4\,\NS_p(f)$.
\end{proposition} 

\begin{proof} 
Since $f_\ell$ and $\wt{f_\ell}$ are $[-1,1]$-valued, we have that 
\begin{align*} 
\Ex_{\bell,\bx}\big[(\wt{f_{\bell}}(\bx)-f_{\bell}(\bx))^2\big] &\le 2 \Ex_{\bell,\bx}\big[|\wt{f_{\bell}}(\bx)-f_{\bell}(\bx)|\big] \\
&= 2 \Ex_{\bell} \Bigg[ \mathop{\Ex_{\bx}}_{\by\sim_p\bx} \big[ | f_{\bell}(\by) - f_{\bell}(\bx) | \big]\Bigg] \\
&= 2 \Ex_{\bell} \Bigg[ \,2\mathop{\Prx_{\bx}}_{\by\sim_p\bx} \big[ f_{\bell}(\by) \ne f_{\bell}(\bx)\big] \,\Bigg]  \\
&= 4 \E_{\bell} \big[ \NS_p(f_{\bell})\big] \\ &=   4\,\NS_p(f,T^{j, p}_f) \le 4\,\NS_p(f),
\end{align*} 
where the final inequality is a consequence of the fact that score is a nonnegative quantity (\Cref{eq:score-non-negative}). 
\end{proof}

\begin{proposition} 
\label{prop:rounding-error}
For any function $g : \bn\to\bits$ and constant $c\in \R$,
\[ \E\big[(g(\bx) - \sign(\E[g]))^2\big] \le 2\E[(g(\bx)-c)^2].  \] 
\end{proposition} 

\begin{proof} 
Let $a \coloneqq \Pr[g(\bx) = 1]$ and assume without loss of generality that $a \geq \frac1{2}$. On one hand, we have that $\E\big[(g(\bx) - \sign(\E[g])^2\big] = \E\big[(g(\bx) - 1)^2\big] =  4(1-a)$.  On the other hand, since 
 \[ \E[(g(\bx)-c)^2] = a(1-c)^2+(1-a)(1+c)^2 \] 
 this quantity is minimized for $c = 2a-1$ and attains value $4a(1-a)$ at this minimum.  Therefore indeed 
 \[ \min_{c\in \R} \big\{ \E[(g(\bx)-c)^2]\big\} = 4a(1-a) \ge 2(1-a) = \lfrac{1}{2}\E\big[(g(\bx) - \sign(\E[g]))^2\big] \] 
and the proposition follows. 
\end{proof} 
 
With~\Cref{prop:NS-of-leaves,prop:rounding-error} in hand, we now bound $\dist(f,T^{j,p}_f)$: 
\begin{align*} 
\dist(f,T^{j,p}_f) &= \Ex_{\bell\sim T^{j,p}_f}\big[ \dist(f_{\bell},\sign(\E[f_{\bell}]))\big]  \\
&= \lfrac1{4} \ds \Ex_{\bell,\bx}\big[ (f_{\bell}(\bx) - \sign(\E[f_{\bell}]))^2 \big]   \\
&\le \lfrac1{2} \ds \Ex_{\bell,\bx} \big[ (f_{\bell}(\bx) - \E[\wt{f_{\bell}}])^2\big] \tag*{(\Cref{prop:rounding-error})} \\
&\le \Ex_{\bell,\bx}\big[ (f_{\bell}(\bx)-\wt{f_{\bell}}(\bx))^2 \big]  +  \Ex_{\bell,\bx}\big[ (\wt{f_{\bell}}(\bx) - \E[\wt{f_{\bell}}])^2 \big] \tag*{(``almost-triangle" inequality)} \\
&\le 4\,\NS_p(f) + \Ex_{\bell}[\Var(\wt{f_{\bell}})]. \tag*{(\Cref{prop:NS-of-leaves})} 
\end{align*}  
By the assumption that we are in Case 2, 
 \begin{align*} 
 \Ex_{\bell}[\Var(\wt{f_{\bell}})] &< 2\Ex_{\bell,\bx}\big[|{T}^\star (\bx)-\wt{f_{\bell}}(\bx)|\big] + 2\,\eps \\
 &\le 2\Big(\Ex_{\bell,\bx}\big[|{T}^\star(\bx) - f_{\bell}(\bx)|\big]  + \Ex_{\bell,\bx}\big[ |f_{\bell}(\bx) - \wt{f_{\bell}}(\bx)| \big] \Big) + 2\,\eps \tag*{(Triangle inequality)}  \\
 &\le O\big(\Ex_{\bell}[\dist(f_{\bell},{T}^\star)]  +\NS_p(f) \big) + 2\,\eps \tag*{(\Cref{prop:NS-of-leaves})} \\
 &\le O \big(\opt_s + \eps + \NS_p(f)\big) \tag*{($\dist(f,{T}^\star)= \opt_s$)} \\
 &\le O(\opt_s + p\log s + \eps) \tag*{(\Cref{prop:NS-of-f})} \\
 &= O(\opt_s + \eps). \tag*{(Our choice of $p = \eps/\log s$)} 
 \end{align*} 
 
 
Summarizing what we have shown through Cases 1 and 2, for all $j\in \N$, we either have 
\[ \NS_p(f,T^{j+1,p}_f) \le  \NS_p(f,T^{j,p}_f)-\frac{\eps^3}{(\log s)^3} \] 
or it must be the case that  $\dist(f,T^{j,p}_f) \le O(\opt_s + \eps)$.  Since $\NS_p(f,T) \in [0,1]$ for all decision trees~$T$, we must fall into the latter case for some $j \le O((\log s)^3/\eps^3)$.  Finally, since $\dist(f, T^{j+1,p}_f) \le \dist(f, T^{j,p}_f)$ for all $j\in \N$, we conclude that $\dist(f,T^{d,p}_f) \le O(\opt_s + \eps)$ for our choice of $d = O((\log s)^3/\eps^3)$, and~\Cref{lemma:BGLT} follows.

\violet{ \begin{remark} 
\label{rem:differences}
A similar structural lemma \violet{was claimed in~\cite{BGLT-NeurIPS1}}.  Their proof, however, relies crucially on a sophisticated result (the ``two function OSSS inequality for semi-metrics" from~\cite{OSSS05}) that was subsequently shown to be false~\cite{Qia21,OD21}.  Our proof of~\Cref{lemma:BGLT} does not use this erroneous result.  

There are also important differences between our setting and that of~\cite{BGLT-NeurIPS1}'s.  \cite{BGLT-NeurIPS1} analyzes a tree, call it~$\Upsilon$, that is analogous to our $T^{d,p}_f$.  Unlike $T^{d,p}_f$, their tree $\Upsilon$ is not necessarily complete: it is iteratively constructed in a top-down manner, where in each iteration the size of the tree grows by one.  In each iteration, the leaf $\ell$ in the current tree with the highest ``value" is replaced with a query the variable of $f_\ell$ with the highest score, where the ``value" of a leaf is defined to be the score of the highest-scoring variable of $f_\ell$ normalized by $\ell$'s depth in the current tree.~\cite{BGLT-NeurIPS1}'s definition of score differs from ours: for their intended application, it was important that the score of a variable can be efficiently estimated to high accuracy from random labeled examples $(\bx,f(\bx))$ where $\bx \sim \bn$ is uniform random;~\Cref{def:score} does not lend itself to such an estimation procedure. 
\end{remark} }

\begin{remark}
\label{remark:BGLT-approx} 
\Cref{lemma:BGLT} concerns the tree $T^{d,p}_f$ as defined in~\Cref{def:top-down-tree}, where each internal node~$v$ of $T^{d, p}_f$ is a query the variable $x_i$ that maximizes $\score_i(f_v,p)$.    For the algorithmic component of~\Cref{thm:reconstruct}, we will need a robust version of~\Cref{lemma:BGLT}. An inspection of its proof shows that the same statement holds for any tree where each internal node $v$ is a query to a variable of {\sl approximately} maximal score, within $\tau \coloneqq O(\eps^3/(\log s)^3)$ of $\max_{j\in [n]} \score_j(f_v,p)$.   Indeed, the only change to the proof will be that for all $j\in \N$, we either have that 
\[ \NS_p(f,T^{j+1,p}_f) \le  \NS_p(f,T^{j,p}_f)-\frac{\eps^3}{(\log s)^3} + \tau \] 
or it must be the case that  $\dist(f,T^{j,p}_f) \le O(\opt_s + \eps)$. Therefore, as long as $\tau \le O(\eps^3/(\log s)^3)$ the conclusion is unaffected.  Similarly, instead of labeling every leaf $\ell$ with $\sign(\E[f_\ell])$, the same conclusion holds if we only require this for leaves $\ell$ such that $|\E[f_\ell]| > \eps$.

\begin{lemma}[Robust version of~\Cref{lemma:BGLT}]
\label{lem:BGLT-approx}
Let $f : \bn\to\bits$ be $\opt_s$-close to a size-$s$ decision tree.  For $d = O((\log s)^3/\eps^3)$, $p = \eps/(\log s)$, and $\tau = O(\eps^3/(\log s)^3)$, let~$T$ be any complete decision tree of depth $d$ satisfying: 
\begin{itemize} 
\item[$\circ$]  At every internal node $v$, the variable $x_i$ that is queried at this node satisfies: 
\[  \score_{i}(f_v,p)  \geq \max_{j \in [n]}\  \{\score_j(f_v,p)\} - \tau.\]
\item[$\circ$]  Every leaf $\ell$ such that $|\E[f_\ell]|> \eps$ is labeled $\sign(\E[f_\ell])$.
\end{itemize} 
Then $\dist(f,T) \le O(\opt_s +\eps)$. 
\end{lemma} 
\end{remark}


\subsection{Algorithmic component of~\Cref{thm:reconstruct}} 
\label{sec:algorithmic} 

\subsubsection{Query-efficient simultaneous score estimation}
\label{sec:score-estimator} 

We begin by designing a query-efficient subroutine that  simultaneously estimates the scores of all~$n$ variables of a function $f$.  The fact that we are able to do so with $O(\log n)$ queries, as opposed to $\Omega(n)$ as would be required by a naive approach, will be a key component in the query efficiency of our reconstructor.  

\begin{theorem}[Score estimator]
    \label{thm:score-estimator}
There is an algorithm which, given query access to a function $f: \bits^n \to \bits$, noise rate $p \in (0,1)$, accuracy parameter $\tau \in (0,1)$, and confidence parameter $\delta \in (0,1)$, for
\[         q = O\left(\frac{\log n + \log(1/\delta)}{\tau^2} \right)
   \]
    makes $O(q)$ queries, runs in $O(q  n)$ time, and returns estimates $\boldeta_1, \ldots, \boldeta_n$ such that, with probability at least $1 - \delta$, satisfies 
\[         \big| \boldeta_i - \score_i(f, p) \big | < \tau \quad \text{for all $i \in [n]$}.
   \]
\end{theorem}

We prove~\Cref{thm:score-estimator} by first giving a $2$-query algorithm, $\UnbiasedEstimator$ (\Cref{fig:UnbiasedEstimator}), that runs in $O(n)$ time and outputs unbiased estimates of all $n$ scores. The algorithm of \Cref{thm:score-estimator} takes the mean of multiple runs of that unbiased estimator, with its guarantees following from a simple concentration bound.

\begin{figure}[h]
  \captionsetup{width=.9\linewidth}
\begin{tcolorbox}[colback = white,arc=1mm, boxrule=0.25mm]
\vspace{3pt} 

$\UnbiasedEstimator(f,p)$:

\begin{enumerate}[align=left]
    \item[\textbf{Input:}] Query access to a function $f: \bits^n \to \bits$ and a noise rate $p \in (0,1)$.
    \item[\textbf{Output:}] Unbiased estimates of $\score_i(f, p)$ for all $i \in [n]$.
\end{enumerate}
\begin{enumerate}
    \item Choose $\bx \in \bits^n$ uniformly at random and generate a $p$-noisy copy $\by$ of $\bx$.
    \item For each $i \in [n]$, return the estimate
    \begin{align*}
        \boldeta_i = \Ind\big[f(\bx) \neq f(\by)\big] \cdot \left (1 - \frac{1}{1 - \frac{p}{2}} \cdot \Ind[\bx_i = \by_i]\right).
    \end{align*}
\end{enumerate}

\end{tcolorbox}
\caption{$\UnbiasedEstimator$ computes unbiased estimates of the scores of all variables of a function $f$.}
\label{fig:UnbiasedEstimator}
\end{figure}

\begin{lemma}[Analysis of $\UnbiasedEstimator$]
    \label{lemma:unbiased-estimates}
    For any $f: \bits^n \to \bits$ and $p \in (0,1)$, let $\boldeta_1, \ldots, \boldeta_n$ be the outputs of $\UnbiasedEstimator(f, p)$. Then 
    \begin{align*}
        \Ex_{\bx, \by}[\boldeta_i] = \score_i(f, p) \quad \text{for all $i\in [n]$.}
    \end{align*}
\end{lemma}
\begin{proof}
    We first note that $\Pr[f(\bx) \neq f(\by)]$ is $\NS_p(f)$ by definition. Therefore, it is enough for us to prove that
    \begin{align}
        \label{eq:unbiased-estimator}
        \Ex_{\bb \in \bits}\big[\NS_{p}(f_{x_i = \bb})\big] = \frac{1}{1 - \frac{p}{2}} \cdot \Prx_{\bx, \by}\left[f(\bx) \neq f(\by) \text{ and } \bx_i = \by_i\right].
    \end{align}
    Given the above equation, the desired result holds by linearity of expectation and the definition of score. Consider the distribution over $(\bx, \by)$ conditioned on the event that $b = \bx_i = \by_i$. That distribution is equivalent to if we picked $\bx$ randomly from the domain of $f_{x_i = b}$ and selected $\by$ by rerandomizing each coordinate in that domain with probability $p$. Therefore,
    \begin{align*}
        \NS_{p}(f_{x_i = b}) &= \Prx_{\bx, \by}[f(\bx) \neq f(\by) \,|\, b= \bx_i = \by_i] \\
        &= \frac{1}{\Prx_{\bx, \by}[b= \bx_i = \by_i]}] \cdot \Prx_{\bx, \by}[f(\bx) \neq f(\by) \text{ and } b= \bx_i = \by_i]. 
    \end{align*}
    We now prove \Cref{eq:unbiased-estimator}:
    \begin{align*}
        \Ex_{\bb \in \bits}\big[\NS_{p}(f_{x_i = \bb})\big]  &= \Ex_{\bb \in \bits}\left[\frac{1}{\ds\Prx_{\bx, \by}[\bb= \bx_i = \by_i]}] \cdot \Prx_{\bx, \by}[f(\bx) \neq f(\by) \text{ and } \bb= \bx_i = \by_i]\right] \\
        &=\frac{1}{\frac{1}{2} \cdot \ds\Prx_{\bx, \by}[ \bx_i = \by_i]}\Ex_{\bb \in\bits}\left[\Prx_{\bx, \by}[f(\bx) \neq f(\by) \text{ and } \bb= \bx_i = \by_i]\right] \\
        &=\frac{1}{\frac{1}{2} \cdot (1 - \frac{p}{2})} \cdot \lfrac{1}{2} \cdot \ds\Prx_{\bx, \by}\left[f(\bx) \neq f(\by) \text{ and } \bx_i = \by_i\right] \\
        &= \frac{1}{1 - \frac{p}{2}} \cdot \Prx_{\bx, \by}\left[f(\bx) \neq f(\by) \text{ and } \bx_i = \by_i\right].
    \end{align*}
\Cref{lemma:unbiased-estimates} then holds by linearity of expectation.
\end{proof}
We now prove \Cref{thm:score-estimator}.
\begin{proof}[Proof of \Cref{thm:score-estimator}]
    The algorithm runs $\UnbiasedEstimator(f, p)$ $q$ times and then outputs the means of each returned estimates. Each estimate from $\UnbiasedEstimator$ is bounded between $-1$ and $1$. By Hoeffding's inequality, for any $i \in [n]$,
    \begin{align*}
        \Pr\big[\big| \boldeta_i - \score_i(f, p) \big | \geq \tau\big] \leq -\exp_e\left(-\frac{q \cdot \tau^2}{2}\right).
    \end{align*}
    For $q$ as in \Cref{thm:score-estimator}, the above probability is at most $\delta/n$. By union bound, all estimates are accurate within $\pm\tau$ with probability at least $1 - \delta$.
    
    Finally, this algorithm uses only $2q = O(q)$ queries. Each run of $\UnbiasedEstimator$ estimator takes $O(n)$ time to construct the query and compute all the estimates, so the entire algorithm takes $O(q n)$ time.
\end{proof}


\subsubsection{Proof of \Cref{thm:reconstruct}}
\label{sec:proof-of-reconstruct} 


We prove \Cref{thm:reconstruct} by providing an algorithm, $\Reconstructor$ (\Cref{fig:BuildStrand}), which assumes query access to a function $f : \bn\to\bits$ and provides fast query access to a tree $T$ meeting the criteria of \Cref{thm:reconstruct}. We build off a simple observation that also underlies \cite{BGLT-NeurIPS2}: to determine the output of a decision tree $T$ on a particular input~$z$, it suffices to build the root-to-leaf path corresponding to~$z$, which  can be exponentially faster than building the entire tree.  Our algorithm is different from~\cite{BGLT-NeurIPS2}'s; as mentioned in the introduction their algorithm is tailored to monotone functions, and is known to fail for non-monotone ones.  We on the other hand leverage the specific structure of $T$ established in~\Cref{sec:structural} together with the query-efficient score estimator from~\Cref{sec:score-estimator} in our design and analysis of $\Reconstructor$. 

$\Reconstructor$ maintains a partial tree $T^\circ$ containing all the root-to-leaf paths in $T$ corresponding to queries received so far. In the pseudocode for $\Reconstructor$, we use the notation $T^\circ_{\mathrm{internal}}(\alpha) \in  [n]  \cup \{\varnothing\}$ to indicate the variable queried in $[n]$ at internal node $\alpha$ of the partial tree $T^\circ$, or $\varnothing$ if that node has not yet been built. Similarly, $T^\circ_{\mathrm{leaf}}(\alpha) \in \{-1,1,\varnothing\}$ indicates the value at leaf $\alpha$ in $T^\circ$, or $\varnothing$ if that value has not yet been decided.

\Crefname{enumi}{Step}{Steps}

\begin{figure}[h]
  \captionsetup{width=.9\linewidth}
\begin{tcolorbox}[colback = white,arc=1mm, boxrule=0.25mm]
\vspace{3pt} 

$\Reconstructor(f, s, \eps, \delta)$:

\begin{enumerate}[align=left]
    \item[\textbf{Input:}] Query access to a function $f: \bits^n \to \bits$, size parameter $s$, error parameter $\eps$, and failure probability $\delta$.
    \item[\textbf{Output:}] Query access to a decision tree $T$ that satisfies $\dist(f, T) \leq O(\opt_s) + \eps$ with probability at least~$1-\delta$.
\end{enumerate}
\begin{enumerate}
    \item Set parameters $d, p$, and $\tau$ as in~\Cref{lem:BGLT-approx}. 
    \item Initialize $T^\circ$ to be the empty partial tree.
    \item Upon receiving an input $z \in \bits^n$: 
    \begin{enumerate}
        \item Initialize $\alpha$ to be the root of $T^\circ$.
        \item Repeat $d$ times.
        \begin{enumerate}
            \item \label{step:score-estimate} If $T^\circ_{\mathrm{internal}}(\alpha)$ is $\varnothing$ use the estimator from \Cref{thm:score-estimator} to compute estimates of $\score_i(f_{\alpha}, p)$ with additive accuracy $\pm \frac{\tau}{2}$ and failure probability $O(\frac{\delta}{2^d})$ for all $i \in [n]$ and set $T^\circ_{\mathrm{internal}}(\alpha)$ to the variable with highest estimated score.
            \item For $i = {T^\circ_{\mathrm{internal}}(\alpha)}$, If $\overline{z}_{i}$ is $1$, set $\alpha$ to its right child. Otherwise, set $\alpha$ to its left child. 
        \end{enumerate}
        \item \label{step:leaf-estimate} If $T^\circ_{\mathrm{leaf}}(\alpha)$ is $\varnothing$, use random samples to estimate $\E[f_{\ell}]$ to additive accuracy $\pm \frac{\eps}{4}$ with failure probability $O(\frac{\delta}{2^d})$ and set $T^\circ_{\mathrm{leaf}}(\alpha)$ to whichever of $\bits$ that estimate is closer to.
        \item Output $T^\circ_{\mathrm{leaf}}(\alpha)$.
    \end{enumerate}
    
\end{enumerate}

\end{tcolorbox}
\caption{$\Reconstructor$ gives efficient query access to a decision tree is close to $f$ with high probability.}
\label{fig:BuildStrand}
\end{figure}

\Cref{thm:reconstruct} follows from the following two lemmas, showing the correctness and efficiency of $\Reconstructor$ respectively.

\begin{lemma}[Correctness of $\Reconstructor$] \label{lemma:reconstructor-correctness}
    For any $f: \bits^n \to \bits$, $s \in \N$, $\eps \in (0, \frac{1}{2})$, $\delta \in (0,1)$, and sequence of inputs $z^{(1)}, \ldots, z^{(m)}\in \bn$, the outputs of $\Reconstructor$ are consistent with some decision tree $T$ where
    \begin{itemize}
        \item[$\circ$] $T$ has size $s^{O((\log s)^2/\eps^3)}$,
        \item[$\circ$] $\dist(T,f) \le O(\opt_s) + \eps$ with probability at least $1-\delta$.
    \end{itemize} 
\end{lemma}
\begin{proof}
   The outputs of $\Reconstructor$ are always consistent with $T^\circ$ and the depth of $T^\circ$ is always capped at $d$. Let $T$ be the tree that $T^\circ$ would be if every $x \in \bits^n$ were given as an input to $\Reconstructor$. Then, $T$ has size at most $2^d = s^{O((\log s)^2/\eps^3)}$, and every output is consistent with $T$.
   
   If all score estimates in \Cref{step:score-estimate} are accurate to $\pm \frac{\tau}{2}$ and expectation estimates is \Cref{step:leaf-estimate} are accurate to $\pm \frac{\eps}{4}$, then $T$ meets the criteria of \Cref{lem:BGLT-approx} and therefore $\dist(T, f) \leq O(\opt_s) +\eps$. The number of time scores are estimated in \Cref{step:score-estimate} is at most the number of internal nodes of~$T$, which is $2^d - 1$. Similarly, the number of expectation estimates in \Cref{step:score-estimate} is at most the number of leaves of $T$, which is $2^d$. By union bound over the possible failures, we see that the failure probability is at most $\delta$.
\end{proof}

\begin{lemma}[Efficiency of $\Reconstructor$]\label{lem:reconstruct-efficiency}
    For any $f: \bits^n \to \bits$, $s \in \N$, $\eps \in (0, \frac{1}{2})$, $\delta \in (0,1)$, particular input $z \in \bits^n$, and
    \begin{align*}
        q = O\left(\frac{(\log s)^9\cdot (\log n) \cdot \log(1/\delta)}{\eps^9} \right),
    \end{align*}
    upon receiving $z$ as input, $\Reconstructor(f,s,\eps,\delta)$ uses $O(q)$ queries and $O(qn)$ time to return an output.
\end{lemma}
\begin{proof}
    On each input, the estimator from \Cref{thm:score-estimator} is used up to $d$ times. Each uses
    \begin{align*}
        q_{\mathrm{inner}} \coloneqq O\left(\frac{\log n + \log(2^d/\delta)}{\tau^2} \right) = O\left(\frac{\log n + d + \log(1/\delta)}{\tau^2} \right)
    \end{align*}
    queries and $O(q_\mathrm{inner}n)$ time. By Hoeffding's inequality, it is sufficient to take
    \begin{align*}
         q_{\mathrm{leaf}} \coloneqq O\left(\frac{\log(2^d/\delta)}{\eps^2} \right) = O\left(\frac{d + \log(1/\delta)}{\eps^2} \right)
    \end{align*}
    random samples in \Cref{step:leaf-estimate}. Therefore, the total number of queries used is
    \begin{align*}
        q &=  q_{\mathrm{inner}} +  q_{\mathrm{leaf}} \\
        &= O\left(\frac{\log n + d + \log(1/\delta)}{\tau^2} \right) + O\left(\frac{d + \log(1/\delta)}{\eps^2} \right) \\
        &=  O\left(\frac{\log n + ((\log s)^3/\eps^3) + \log(1/\delta)}{\eps^6 /( \log s)^6} + \frac{((\log s)^3/\eps^2) + \log(1/\delta)}{\eps^2} \right)\\
        &= O\left(\frac{(\log s)^9\cdot (\log n) \cdot \log(1/\delta)}{\eps^9} \right).
    \end{align*}
    The time to prepare all queries is $O(q n)$, and all other computation is asymptotically faster.
\end{proof}

\begin{remark}[Local reconstruction]
    \label{remark: local reconstruction}
    We remark that our reconstruction algorithm can be made local in the sense of \cite{SS10}. They define a reconstruction algorithm, $\mathcal{A}$, to be local, if the output of $\mathcal{A}$ on some input $z$ is a \emph{deterministic} and easy to compute function of $z$ and some small random string $\rho$. This allows queries to the reconstructor to be answered in parallel, as long as the random string $\rho$ is shared. To make our reconstructor local, we note that the only place randomness is used is in generating samples consistent with some restriction $\alpha$. We can set $\rho$ to be $n$ bits per sample the constructor might wish to generate. Since the total number of samples the reconstructor needs per input is $\poly(\log s, 1/\eps, \log(1/\delta))\cdot \log n$, we have
    \begin{align*}
        |\rho| = \poly(\log s, 1/\eps, \log(1/\delta))\cdot n\log n
    \end{align*}
    On a particular input, the local reconstructor starts with $T^\circ$ being the empty tree. Whenever it wishes to produce a random sample consistent with $\alpha$, it sets the $\bx \in \bits^n$ to be next $n$ bits of $\rho$ and then uses $\bx_{\alpha}$ for the sample. It's easy to see that this algorithm will keep $T^\circ$ consistent between different runs because it will always compute the same variable as having the highest score given some restriction. Furthermore, the analysis goes through without issue. The only difference between this analysis and one where fresh random bits are used to for each sample is that the queries of different paths may be correlated. In our proof of \Cref{lemma:reconstructor-correctness}, we use a union bound to ensure all estimates obtained through sampling are accurate, and that union bound holds regardless of whether those estimates are independent.
\end{remark}

\subsection{Proof of~\Cref{cor:main}} 
\label{sec:proof-of-tester}
In this section we derive \Cref{cor:main} as a simple consequence of~\Cref{thm:reconstruct}. The connection between reconstruction and tolerant testing has been noted in other works (see e.g.~\cite{CGR13, B08}); we provide a proof here for completeness.

\begin{theorem}[\Cref{cor:main} restated]
    There is an algorithm which, given query access to $f : \bn \to\bits$ and parameters $s\in\N$ and $\eps, \delta \in (0,1)$, runs in $\poly(\log s,1/\eps)\cdot n\log n\cdot  \log(1/\delta)$ time, makes $\poly(\log s,1/\eps) \cdot \log n \cdot \log(1/\delta)$ queries to $f$, and 
    \begin{itemize}
    \item[$\circ$] Accepts w.p.~at least $1 - \delta$ if $f$ is $\eps$-close to a size-$s$ decision tree; 
    \item[$\circ$] Rejects w.p.~at least $1 - \delta$ if $f$ is $\Omega(\eps)$-far from size-$s^{O((\log s)^2/\eps^3)}$ decision trees.
    \end{itemize} 
\end{theorem}
\begin{proof}
The algorithm chooses $m$ uniform random inputs, $\bx^{(1)}, \ldots, \bx^{(m)} \sim \bits^n$ where $m = O(\log(1/\delta)/\eps^2)$. Let $\bb^{(1)}, \ldots, \bb^{(m)} \in \bits$ be the output of $\Reconstructor(f, s, \eps, \delta)$. The tester rejects if $\E_{\bi \in [m]}[f(\bx^{(\bi)}) \neq \bb^{(\bi)}] > \Omega(\eps)$ and accepts otherwise.
    
    First, we consider the case where $f$ is $\eps$-close to a size-$s$ decision tree (i.e.~$\opt_s \le \eps$). By \Cref{lemma:reconstructor-correctness}, with probability at least $1 - \delta$ the outputs of $\Reconstructor$ are consistent with a tree, $T$, satisfying $\dist(T, f) \leq O(\eps)$. By Hoeffding's inequality,
    \[ 
        \Prx_{\bx^{(1)}, \ldots, \bx^{(m)}}\left[\Ex_{\bi \in [m]}\big[f(\bx^{(\bi)}) \neq \bb^{(\bi)}\big] > \Omega(\eps)\right] \leq \exp(-2 m \eps^2) \le \delta.\] 
    By a union bound, the tester rejects with probability at most $\delta + \delta = 2\delta$.
    
    We next consider the case where $f$ is $\Omega(\eps)$-far from size-$s^{O((\log s)^2/\eps^3)}$ decision trees. By \Cref{lemma:reconstructor-correctness} it is guaranteed to be consistent. A similar argument to the first case shows that the probability of acceptance is at most $\delta + \exp(-2 m \eps^2) = 2\delta$.    Finally, the efficiency of this tester is a consequence of \Cref{lem:reconstruct-efficiency} and our choice of $m = O(\log(1/\delta)/\eps^2).$ 
\end{proof}





\Crefname{enumi}{Item}{Item}

\section{Proof of \Cref{cor:reconstruct-degree}}

We first restate~\Cref{thm:reconstruct} with decision tree {\sl depth} instead of {\sl size} as the complexity measure: 

\begin{theorem}[\Cref{thm:reconstruct} in terms of decision tree depth]
\label{thm:reconstruct-depth}
There is a randomized algorithm which, 
given query access to $f : \bn\to\bits$ and parameters $d\in \N$ and $\eps\in (0,1)$, provides query access to a fixed decision tree $T$ where 
\begin{itemize}
\item[$\circ$] $T$ has depth $O(d^3/\eps^2)$,
\item[$\circ$] $\dist(T,f) \le O(\opt_d) + \eps$ w.h.p., where $\opt_d$ denotes the distance of $f$ to the closest depth-$d$ decision tree.
\end{itemize} 
Every query to $T$ is answered in $\poly(d,1/\eps)\cdot n\log n$ time and with $\poly(d,1/\eps)\cdot \log n$ queries to~$f$. 
\end{theorem}

To see that our proof of~\Cref{thm:reconstruct} also establishes~\Cref{thm:reconstruct-depth}, we use the fact that every depth-$d$ decision tree has size $\le 2^d$, and recall that the tree $T$ that the algorithm of~\Cref{thm:reconstruct} provides query access to is a complete tree and hence has depth logarithmic in its size.  

Decision tree depth and Fourier degree of boolean functions are known to be polynomially related: 
 
 \begin{fact}[Decision tree depth vs.~Fourier degree~\cite{Mid04,Tal13}] 
 \label{fact:depth-vs-degree} 
 For $g : \bn\to\bits$ let $\deg(g)$ denote $g$'s Fourier degree and $D(g)$ denote the depth of the shallowest decision tree that computes~$g$.  Then $\deg(g) \le D(g)$ and $D(g) \le \deg(g)^3$.  
 \end{fact} 
 
 We first observe \Cref{thm:reconstruct-depth} and~\Cref{fact:depth-vs-degree} already gives a quantitatively weaker version of~\Cref{cor:reconstruct-degree} where $g$ has degree $O(d^{\,9}/\eps^2)$.  To see this $f : \bn\to\bits$ be $\opt_d$-close to a degree-$d$ function $h : \bn\to\bits$.  By~\Cref{fact:depth-vs-degree}, $D(h) \le \deg(h)^3$, and so the algorithm of~\Cref{thm:reconstruct-depth} provides query access to a decision tree $T : \bn\to\bits$ that is $(O(\opt_d) + \eps)$-close to $f$ and where the depth of $T$ is $O(D(h)^3/\eps^2) = O(\deg(h)^9/\eps^2)$.  Applying~\Cref{fact:depth-vs-degree} again, we conclude that $\deg(T) \le D(T) \le O(\deg(h)^9/\eps^2)$. 
 
 To obtain the sharper bound of $O(\deg(h)^7/\eps^2)$, we observe that the proof of~\Cref{lemma:BGLT} in fact bounds the depth of $T$ by $O(D(h)^2\, \Inf(h)/\eps^2)$.  
 Influence and degree of boolean functions are related via the following basic fact (see e.g.~\cite[Theorem 37]{ODbook}): 
 \begin{fact}
 For all $h : \bn\to\bits$, we have $\Inf(h) \le \deg(h)$.  
 \end{fact} 
Therefore, we can bound the degree of $T$ by $O(D(h)^2\,\Inf(h)/\eps^2) \le O(\deg(h)^7/\eps^2)$.  

Guarantees for the other measures listed in~\Cref{table:other-measures} follow from similar calculations and known quantitative relationships between these measures and decision tree complexity; the current best bounds are summarized in Table 1 of~\cite{ABDKRT20}.

%
%
%
%
%
%


\section{Proof of~\Cref{thm:lower-bound-intro}}

In this section, we prove the following theorem:

\begin{theorem}[Tolerant testing of DTs $\Rightarrow$ Proper learning of DTs]
    \label{thm:lower-bound}
    Let $c > 0$ be an absolute constant and $\mathcal{A}$ be an algorithm with the following guarantee. Given query access to $f: \bn \to \bits$ and parameters $s \in \N $ and $\eps \in (0,1)$, the algorithm $\mathcal{A}$: 
    \begin{itemize}
        \item[$\circ$] Accepts w.h.p.~if $f$ is $\eps$-close to a size-$s$ decision tree; 
        \item[$\circ$] Rejects w.h.p.~if $f$ is $(c\eps)$-far from all size-$s$ decision trees.
    \end{itemize}
    Then there is an algorithm $\mathcal{B}$ with the following guarantee. Given parameters $s' \in \N$ and $\eps' \in (0,1)$, and query access to a function $g : \bn\to\bits$ that is computed by a size-$s'$ decision tree, $\mathcal{B}$ makes $\poly(s', n, 1/\eps')$ calls to $\mathcal{A}$, each with parameters $s \leq s'$ and $\eps \ge  \poly(1/s',\eps')$, and produces a decision tree which is $\eps'$-close to $g$ with high probability. Furthermore, the auxiliary computation that $g$ does takes time $\poly(n, s', 1/\eps')$.
\end{theorem}

\Cref{thm:lower-bound-intro} follows as a special case of \Cref{thm:lower-bound}. 

We prove~\Cref{thm:lower-bound} in two steps:

\begin{enumerate}
    \item A tolerant tester implies an algorithm for estimating the distance of any function to the class of size-$s$ decision trees. This is well known~\cite{PRR06} and applies to any function class, not just decision trees.
    \item An algorithm for estimating distance to decision trees implies a proper learner for decision trees. Here, we take advantage of the structure of decision trees.
\end{enumerate}

For a function $g : \bn\to\bits$ and $s\in \N$, we write $\opt_s(g)$ to denote the distance of $g$ to the closest size-$s$ decision tree. 

\begin{lemma}[Tolerant testing $\Rightarrow$ distance estimation~\cite{PRR06}] 
    \label{lem:distance-estimator}
    Let $c$ and $\mathcal{A}$ be as in \Cref{thm:lower-bound}. There exists an estimator $\mathcal{E}$ with the following guarantee. Given query access to $g: \bn \to \bits$ and parameters $s' \in \N$ and $\gamma \in (0,1)$, the estimator $\mathcal{E}$ makes $c/\gamma$ calls to $\mathcal{A}$ and returns an $\boldeta$ that with high probability satisfies
\[         \boldeta \leq  \opt_s(g) \leq c \cdot \boldeta + \gamma.
   \]
    Furthermore, the auxiliary computation of $g$ takes time $O(c/\gamma)$.
\end{lemma}

\Crefname{enumi}{Step}{Steps}

\begin{proof}
The algorithm $\mathcal{E}$ runs $\mathcal{A}$ with $\eps = \frac{\gamma}{c},  \frac{2\gamma}{c}, \frac{3\gamma}{c},\ldots, 1$, and sets $\boldeta$ to be the largest $\eps$ for which $\mathcal{A}(g, s, \eps)$ rejects. 
Since $\mathcal{A}(g, s, \boldeta)$ rejected, 
\[         \boldeta < \opt_s(g)\] 
with high probability.    Furthermore, since $\mathcal{A}(g, s, \boldeta + \frac{\gamma}{c})$ accepted,
    \begin{align*}
        \opt_s(g) &< c \cdot \big(\boldeta + \lfrac{\gamma}{c}\big) \\
         &= c \cdot \boldeta + \gamma
    \end{align*}
with high probability.  Finally, we note that $\mathcal{E}$ indeed makes $c/\gamma$ calls to $\mathcal{A}$, and aside from those calls, it only needs to make a single pass over the output of those calls and return the largest $\eps$ that led to a rejection, which takes time $O(c/\gamma)$.
\end{proof}

We are now ready to state our algorithm, $\BuildDT$ (\Cref{fig:BuildDT}), for properly learning size-$s'$ decision trees.  $\BuildDT$ will additionally take in a depth parameter $d$ that will facilitate our analysis of it (looking ahead, $d$ will be chosen to be $O(\log(s'/\eps'))$ in our proof of~\Cref{thm:lower-bound}).
\begin{figure}[h]
  \captionsetup{width=.9\linewidth}
\begin{tcolorbox}[colback = white,arc=1mm, boxrule=0.25mm]
\vspace{3pt} 

$\BuildDT(f,s,d,\gamma)$:

\begin{enumerate}[align=left]
    \item[\textbf{Input:}] Query access to $f: \bn \to \bits$, parameters $s,d\in \N$ and $\gamma \in (0,1)$. 
    \item[\textbf{Output:}] A size-$s$ depth-$d$ decision tree $T$.
\end{enumerate}
\begin{enumerate}
    \item \label{alg-base-case} If $s = 1$ or $d = 0$, return $\sign(\E[f])$.
    \item For each $i \in [n]$ and integers $s_0, s_1 \geq 1$ satisfying $s_0 + s_1 = s$:
    \begin{enumerate}
        \item \label{subroutine} Use $\mathcal{E}$ from~\Cref{lem:distance-estimator} to obtain estimates $\boldeta(x_i=0,s_0)$ and $\boldeta(x_i=1,s_1)$ that satisfy:
        \begin{align*}
           \boldeta(x_i=0,s_0) \leq \opt_{s_0}(f_{x_i = 0}) \leq c \cdot \boldeta(x_i=0,s_0) + \gamma; \\
           \boldeta(x_i=1,s_1) \leq \opt_{s_1}(f_{x_i = 1}) \leq c \cdot \boldeta(x_i=1,s_1) + \gamma.
        \end{align*}
        \item Store $\mathrm{error}(i, s_1, s_2) \leftarrow \frac{1}{2} \big(\boldeta(x_i=0,s_0) + \boldeta(x_i=1,s_1)\big)$.
    \end{enumerate}
    \item \label{minerror} Let $(i^\star, s_0^\star, s_1^\star)$ be the tuple that minimizes $\mathrm{error}(i, s_0, s_1)$. Output the tree with $x_{i^\star}$ as its root, $\BuildDT(f_{x_{i^\star} = 0}, s_0^\star, d-1,\gamma)$ as its left subtree, and $\BuildDT(f_{x_{i^\star} = 1}, s_1^\star, d-1,\gamma)$ as its right subtree.
\end{enumerate}

\end{tcolorbox}
\caption{$\BuildDT$ computes a size-$s$ depth-$d$ decision tree that approximates a target function $f: \bn \to \bits$.} 
\label{fig:BuildDT}
\end{figure}


\begin{lemma}[Error of $\BuildDT$]
    \label{lemma:build-dt-error}
    For all functions $f: \bn \to \bits$ and parameters $s,d \in \N$ and $\gamma\in (0,1)$, the algorithm $\BuildDT(f,s,d,\gamma)$ outputs a decision tree $T$ satisfying
    \begin{align}
        \label{eq:buildDT-accuracy}
        \dist(T, f) \leq c^d \cdot \opt_s(f) + \gamma \cdot \frac{c^d - 1}{c-1} + \frac{s}{2^{d+2}}.
    \end{align}
\end{lemma}
\begin{proof}
    We proceed by induction on $s$ and $d$. If $s = 1$, then in \Cref{alg-base-case}, \BuildDT\ outputs the best decision tree of size $1$. Therefore, $\dist(T,f) \leq \opt_s(f)$, satisfying \Cref{eq:buildDT-accuracy}. If $d = 0$ and $s \geq 2$, then $\frac{s}{2^{d+2}} \geq \frac{2}{4} = \frac1{2}$. Furthermore, in \Cref{alg-base-case}, \BuildDT\ always outputs a tree with error at most $\frac{1}{2}$. Therefore,
\[        \dist(T, f) \leq  \frac{s}{2^{d+2}} \leq c^d \cdot \opt_s(f) + \gamma \cdot \frac{c^d - 1}{c-1} + \frac{s}{2^{d+2}}.\] 
    Finally, we consider the case where $d \geq 1$ and $s \geq 2$. Let $T_{\opt}$ be the size-$s$ decision tree that is $\opt_s(f)$ close to $f$. Let $x_{i_{\opt}}$ the root of $T_{\opt}$, and $s_{0, \opt}$, $s_{1, \opt}$ the sizes of the left and right subtrees of $T_{\opt}$ respectively. Since the estimates computed in \Cref{subroutine} are underestimates of or equal to the true error (i.e.~$\error(i_{\opt}, s_{0, \opt}, s_{1, \opt}) \leq \opt_s(f)$), and 
since $i^\star, s_0^\star, s_1^\star$ are chosen in~\Cref{minerror} to minimize the estimated error, we have 
 \[ 
        \error(i^\star, s_0^\star, s_1^\star) \leq \error(i_{\opt}, s_{0, \opt}, s_{1, \opt}) \leq \opt_s(f).
 \] 
  Finally, we bound $\dist(T, f)$. Let $T_0$ and $T_1$ be the left and right subtrees of $T$. Then,
    \begin{align*}
        \dist(T, f) &= \lfrac{1}{2} \big(\dist(T_0, f_{x_{i^\star} = 0}) + \dist(T_1, f_{x_{i^\star} = 1}) \big) \\
            &\leq \lfrac{1}{2} \ds \Big(c^{d-1} \cdot \opt_{s_0^\star}(f_{x_{i^\star} = 0}) + \gamma \cdot \frac{c^{d-1} - 1}{c-1} + \frac{s_0^\star}{2^{d+1}} \\
            & \ \ \ \ +   c^{d-1} \cdot \opt_{s_1^\star}(f_{x_{i^\star} = 1}) + \gamma \cdot \frac{c^{d-1} - 1}{c-1} + \frac{s_1^\star}{2^{d+1}} \Big) \tag{Inductive hypothesis} \\
            &= c^{d-1} \cdot \lfrac{1}{2}\, \ds \big(\opt_{s_0^\star}(f_{x_{i^\star} = 0}) + \opt_{s_1^\star}(f_{x_{i^\star} = 1})\big) + \gamma \cdot \frac{c^{d-1} - 1}{c-1} +  \frac{s_0^\star + s_1^\star}{2^{d+2}} \\
            &\leq c^{d-1} \cdot \lfrac{1}{2} \ds \big((c \cdot \boldeta(x_{i^\star}=0,s_0^\star) + \gamma) + (c \cdot \boldeta(x_{i^\star} = 1, s_1^\star) + \gamma)\big) + \gamma \cdot \frac{c^{d-1} - 1}{c-1} + \frac{s}{2^{d+2}}\\
            &= c^{d} \cdot \lfrac{1}{2} \ds \,\big(\boldeta(x_{i^\star} = 0, s_0^\star) + \boldeta(x_{i^\star} = 1, s_1^\star)\big) + c^{d-1} \cdot \gamma + \gamma \cdot \frac{c^{d-1} - 1}{c-1} + \frac{s}{2^{d+2}} \\
            &= c^{d} \cdot  
            \error(i^\star, s_0^\star, s_1^\star) + \gamma \cdot \frac{c^{d-1}(c-1) + c^{d-1} - 1}{c-1} + \frac{s}{2^{d+2}} \\
            &\leq c^d\cdot \opt_s(f) + \gamma \cdot \Big(\frac{c^d  - 1}{c-1}\Big) + \frac{s}{2^{d+2}}.
    \end{align*}
    The desired result holds by induction.
\end{proof}

For readability, \Cref{lemma:build-dt-error} assumes that $\BuildDT$ is able to compute $\round(\E[f])$ in \Cref{alg-base-case}. To make $\BuildDT$ efficient, we would only estimate $\E[f]$ by querying $f$ on uniform random inputs $\bx \in \bn$. If those estimates are computed to accuracy $\eps'$, then each leaf of our tree can have up to $\eps'$ additional error.  This is not an issue since it increases the total error of $T$, which is simply the average of the error at each leaf, by only $\eps'$.

Finally, we prove \Cref{thm:lower-bound}:
\begin{proof}[Proof of \Cref{thm:lower-bound}]
    Our goal is to properly learn a size-$s'$ decision tree $g: \bn \to \bits$ to accuracy $\eps'$. To do so, we run $\BuildDT(g, s',d,\gamma)$,  with $d$ set to
    \begin{align*}
        d = \log(s'/\eps') - 1,
    \end{align*}
    and $\gamma$ set to
\[      \gamma = \frac{\eps'}{2 \cdot \max(2, c)^d} = \frac{\eps'}{2} \cdot (2^{-d})^{\max(1, \log c)}
            = \frac{\eps'}{2} \cdot \left(\frac{\eps'}{s'}\right)^{\max(1, \log c)}.
    \] 
    By \Cref{lemma:build-dt-error}, for $T$ the tree $\BuildDT$ outputs,
    \begin{align*}
        \dist(T,g) &\leq c^d \cdot \opt_{s'}(g) + \gamma \cdot \frac{c^d - 1}{c-1} + \frac{s'}{2^{d+2}} \\
        &\leq 0 + \frac{\eps'}{2 \cdot \max(2, c)^d} \cdot \max(2, c)^d + \frac{s'}{2^{\log(s'/\eps') + 1}} \\
        & \leq \frac{\eps'}{2} + \frac{\eps'}{2} = \eps'.
    \end{align*}
    Hence, $\BuildDT$ produces the desired output. We next argue that it is efficient. During the recursion, $\BuildDT$ is called at most $s'$ times in total. Each such call makes $O(ns')$ calls to $\mathcal{E}$. By \Cref{lem:distance-estimator}, those calls to $\mathcal{E}$ each make $c/\eps$ calls to $\mathcal{A}$. Hence, the total number of calls to $\mathcal{A}$ is
    \begin{align*}
O\left(\frac{n(s')^2}{\gamma}\right) = O\left(\frac{n (s')^2}{\eps'} \cdot \left(\frac{s'}{\eps'}\right)^{\max(1 ,\log c)}\right) = \poly(n,s',1/\eps').
    \end{align*}
    The total auxiliary computation of $\BuildDT$ is bounded by the same quantity. Finally, each call to $\mathcal{A}$ is made with parameters $s$ and $\eps$ where $s \leq s'$ and $\eps = \frac{\eps'}{2} \cdot \left(\frac{\eps'}{s'}\right)^{\max(1, \log c)} \ge \poly(1/s',\eps')$.
\end{proof}

\Crefname{enumi}{Item}{Item}

\section*{Acknowledgements} 

We thank the anonymous reviewers for their thoughtful comments and feedback. 

Guy and Li-Yang are supported by NSF CAREER Award 1942123. Jane is supported by NSF Award CCF-2006664. 

\bibliography{most-influential}{}
\bibliographystyle{alpha}

\end{document}